\documentclass[10pt,journal,epsfig]{IEEEtran}

\usepackage[dvips]{graphicx}
\usepackage{graphicx}
\usepackage{amssymb}
\usepackage{cite}
\usepackage{subfigure}
\usepackage{amsmath}

\usepackage{algorithm}
\usepackage{algorithmic}

\begin{document}

\title{Low-Rank Covariance-Assisted Downlink Training and Channel Estimation for FDD Massive MIMO Systems}

\author{Jun Fang, Xingjian Li, Hongbin Li,~\IEEEmembership{Senior
Member,~IEEE}, and Feifei Gao
\thanks{Jun Fang, and Xingjian Li are with the National Key Laboratory
of Science and Technology on Communications, University of
Electronic Science and Technology of China, Chengdu 611731, China,
Email: JunFang@uestc.edu.cn}
\thanks{Hongbin Li is
with the Department of Electrical and Computer Engineering,
Stevens Institute of Technology, Hoboken, NJ 07030, USA, E-mail:
Hongbin.Li@stevens.edu}
\thanks{Feifei Gao is with the
Institute of Information Processing, Department of Automation,
Tsinghua University, Beijing 100084, China, Email:
feifeigao@tsinghua.edu.cn}
\thanks{This work was supported in part by the National Science
Foundation of China under Grant 61522104, and the National Science
Foundation under Grant ECCS-1408182 and Grant ECCS-1609393.}}

\maketitle


\begin{abstract}
We consider the problem of downlink training and channel
estimation in frequency division duplex (FDD) massive MIMO
systems, where the base station (BS) equipped with a large number
of antennas serves a number of single-antenna users
simultaneously. To obtain the channel state information (CSI) at
the BS in FDD systems, the downlink channel has to be estimated by
users via downlink training and then fed back to the BS. For FDD
large-scale MIMO systems, the overhead for downlink training and
CSI uplink feedback could be prohibitively high, which presents a
significant challenge. In this paper, we study the behavior of the
minimum mean-squared error (MMSE) estimator when the channel
covariance matrix has a low-rank or an approximate low-rank
structure. Our theoretical analysis reveals that the amount of
training overhead can be substantially reduced by exploiting the
low-rank property of the channel covariance matrix. In particular,
we show that the MMSE estimator is able to achieve exact channel
recovery in the asymptotic low-noise regime, provided that the
number of pilot symbols in time is no less than the rank of the
channel covariance matrix. We also present an optimal pilot design
for the single-user case, and an asymptotic optimal pilot design
for the multi-user scenario. Lastly, we develop a simple
model-based scheme to estimate the channel covariance matrix,
based on which the MMSE estimator can be employed to estimate the
channel. The proposed scheme does not need any additional training
overhead. Simulation results are provided to verify our
theoretical results and illustrate the effectiveness of the
proposed estimated covariance-assisted MMSE estimator.
\end{abstract}


\begin{keywords}
Massive MIMO systems, downlink training and channel estimation,
channel covariance matrix, low rank structure, MMSE estimator.
\end{keywords}


\section{Introduction}
Massive multiple-input multiple-output (MIMO), also known as
large-scale or very-large MIMO, is a promising technology to meet
the ever growing demands for higher throughput and better
quality-of-service of next-generation wireless communication
systems \cite{RusekPersson13,LarssonEdfors14}. Massive MIMO
systems are those that are equipped with a large number of
antennas at the base station (BS) simultaneously serving a much
smaller number of single-antenna users sharing the same
time-frequency slot. By exploiting the asymptotic orthogonality
among channel vectors associated with different users, massive
MIMO systems can achieve almost perfect inter-user interference
cancelation with a simple linear precoder and receive combiner
\cite{Marzetta10}, and thus have the potential to enhance the
spectrum efficiency by orders of magnitude. In addition to higher
throughput, massive MIMO systems can also improve the energy
efficiency and enable the use of inexpensive, low-power components
\cite{NgoLarsson13}.



To reach the full potential of massive MIMO, accurate downlink
channel state information (CSI) is required at the base station
(BS) for precoding and other operations. Downlink channel
estimation for massive MIMO systems has been extensively studied
over the past few years. Most of existing studies, e.g.
\cite{RusekPersson13,Marzetta10,YinGesbert13,MullerCottatellucci14}
assume a time division duplex (TDD) mode in which channel
reciprocity between opposite links (downlink and uplink) can be
exploited to facilitate the acquisition of the downlink CSI at the
BS. Nevertheless, it was pointed out that the reciprocity of the
wireless channel may not hold exactly due to calibration errors in
the downlink/uplink RF chains \cite{GueyLarsson04}. Also, it is
noted that current wireless cellular systems are still primarily
based on the frequency division duplex (FDD). To make the massive
MIMO technique backward compatible with current systems, it is of
great necessity to study downlink channel estimation for FDD
massive MIMO systems.



For FDD systems, the reciprocity between downlink and uplink
channels no longer holds. To obtain the channel state information
at the transmitter (CSIT), the BS needs to transmit training
signals to users, and each user, after acquiring the downlink CSI
through the training phase, feeds back the CSI to the BS. The
problem lies in that the required amount of overhead for downlink
training grows linearly with the number of transmit antennas at
the BS. This may not be an issue for conventional MIMO scenarios
with only a small number of antennas. However, for massive MIMO
systems where the number of transmit antennas at the BS is large,
the overhead for the downlink training and uplink feedback could
become prohibitively high. Therefore reducing the overhead for
downlink training and uplink CSIT feedback has been a central
issue in FDD massive MIMO systems. A multitude of efforts has been
directed towards this goal over the past few years, e.g.
\cite{RaoLau14,GaoDai15,XieGao15,AdhikaryNam13,SunGao15,ChoiLove14,NohZoltowski14,JiangMolisch15}.
Specifically, in \cite{RaoLau14,GaoDai15,XieGao15}, the sparsity
of the channel on the virtual angular domain has been leveraged to
formulate downlink channel estimation as a compressed sensing
problem, based on which the overhead for downlink training and
uplink feedback can be substantially reduced. Recent experiments
and studies (e.g. \cite{ZhouHerdin07,YinGesbert13}) show that for
a typical cellular configuration with a tower-mounted BS, the
angular spread of the incoming/outgoing rays at the BS is usually
small, and as a result, the channel has a sparse or an approximate
sparse representation on the virtual angular domain.





Besides compressed sensing-based techniques
\cite{RaoLau14,GaoDai15}, another line of research approaches the
overhead reduction issue for FDD massive MIMO by implicitly or
explicitly exploiting the low-rank structure of the channel
covariance matrix, e.g.
\cite{ChoiLove14,NohZoltowski14,AdhikaryNam13,SunGao15,JiangMolisch15}.
Low-rank channel covariance matrix also arises as a result of a
small angular spread of the incoming/outgoing rays at the BS. Due
to the narrow angular spread, different paths between the BS and
the user are highly correlated, and consequently, the channel
covariance matrix has a low-rank or an approximate low-rank
structure with only a few dominant eigenvectors
\cite{ShepardYu12,YinGesbert13}. In \cite{BjornsonOttersten10}, it
was shown that even for conventional MIMO scenarios, the dimension
of the optimal pilot can be reduced if there are only a few
dominant eigenvectors associated with the channel covariance
matrix. Covariance-aided pilot design was also considered in
\cite{ChoiLove14,NohZoltowski14} for FDD massive MIMO systems,
where open-loop and closed-loop training strategies were developed
to reduce the overhead of the downlink training phase by
exploiting the spatial correlation as well as the temporal
correlation of the channel. In \cite{AdhikaryNam13,SunGao15}, the
dimensionality of the effective channels is reduced via a
prebeamforming matrix that depends only on the channel
second-order statistics (i.e. channel covariance matrix), based on
which a joint spatial division and multiplexing (JSDM) scheme
\cite{AdhikaryNam13} and a beam division multiple access scheme
\cite{SunGao15} were proposed to achieve significant savings in
both the downlink training and the CSIT uplink feedback.

In this paper, we continue the direction of covariance-aided
downlink training and channel estimation for FDD massive MIMO
systems. Specifically, we study the asymptotic behavior of the
minimum mean-squared error (MMSE) estimator when the channel
covariance matrix has a low-rank structure. Our theoretical
results reveal that with a low-rank channel covariance matrix, the
MMSE estimator employing a random (not necessarily optimal) pilot
can obtain a perfect channel recovery in the limit of vanishing
noise, provided that the length of the pilot (i.e. the number of
symbols in time) is no less than the rank of the covariance
matrix. We also examine asymptotically optimal pilot design for
the multiple-user scenario. An overlayed training strategy similar
to the JSDM scheme is proposed and shown to be asymptotically
optimal in terms of estimation errors when users have mutually
non-overlapping angles of arrival (AoAs). The optimal design
suggests that the minimum MSE can be achieved as long as the
length of pilot is no less than the rank of the channel covariance
matrix. In addition, based on the one-ring model, we develop a
simple model-based scheme to estimate the channel covariance
matrix. The proposed scheme does not require any additional
training overhead. Simulation results show that the proposed
estimated covariance-assisted MMSE estimator achieves a
substantial performance improvement over the compressed
sensing-based methods.


The rest of this paper is organized as follows. In Section
\ref{sec:problem-formulation}, we introduce the system model and
basic assumptions. The asymptotic behavior of the MMSE estimator
in the limit of vanishing noise is examined in Section
\ref{sec:low-noise-MMSE}. An optimal pilot design for the
single-user scenario and an asymptotic optimal pilot design for
the multi-user scenario are studied in Sections
\ref{sec:pilot-design-single-user} and
\ref{sec:pilot-design-multiuser}, respectively. In Section
\ref{sec:estimated-covariance-MMSE}, we develop a simple
model-based scheme to estimate the channel covariance matrix, and
construct a MMSE estimator to estimate the channel. Simulation
results are provided in Section \ref{sec:simulation}, followed by
concluding remarks in Section \ref{sec:conclusion}.

\begin{figure}[!t]
\centering
\includegraphics[width=8cm]{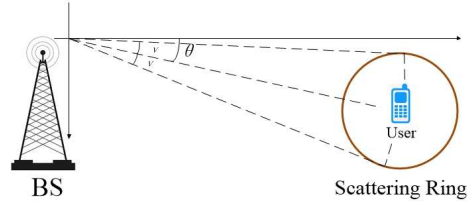}
\caption{Schematic for the one-ring model.} \label{fig6}
\end{figure}

\section{System Model and Problem Formulation} \label{sec:problem-formulation}
We consider the problem of downlink training and channel
estimation in a frequency division duplex (FDD) massive MIMO
system, where the base station (BS) equipped with a large number
of antennas serves a number of single-antenna users
simultaneously. To simplify our problem, we consider the
single-user scenario. The extension of our results to the
multi-user scenario is straightforward, and the pilot design for
the multi-user case will be discussed in Section
\ref{sec:pilot-design-multiuser}. We assume the channel
$\boldsymbol{h}\in\mathbb{C}^{M}$ is a flat Rayleigh fading
channel under a narrowband assumption, where $M$ denotes the
number of transmit antennas at the BS. The extension to the
wideband frequency-selective channel is straightforward when an
OFDM transmission scheme is adopted. The signal received by the
user can be expressed as
\begin{align}
y_t=\boldsymbol{x}_t^T\boldsymbol{h}+w_t \quad \forall t=1,\ldots,
T \label{received-signal}
\end{align}
where $\boldsymbol{x}_t\in\mathbb{C}^{M}$ is the transmitted pilot
symbol vector at time $t$, and $w_t$ denotes the additive white
Gaussian noise with zero mean and variance $\sigma^2$. Define
$\boldsymbol{y}\triangleq
[y_1\phantom{0}y_2\phantom{0}\ldots\phantom{0}y_T]^T$,
$\boldsymbol{X}\triangleq
[\boldsymbol{x}_1\phantom{0}\boldsymbol{x}_2\phantom{0}\ldots\phantom{0}\boldsymbol{x}_T]^T$,
and $\boldsymbol{w}\triangleq
[w_1\phantom{0}w_2\phantom{0}\ldots\phantom{0}w_T]^T$. The data
model (\ref{received-signal}) can be rewritten as
\begin{align}
\boldsymbol{y}=\boldsymbol{X}\boldsymbol{h}+\boldsymbol{w}
\label{data-model}
\end{align}

In this paper, we consider the classical one-ring model that has
been widely adopted (e.g.
\cite{YinGesbert13,GaoDai15,AdhikaryNam13}) to characterize the
massive MIMO channel, where the BS is assumed to be located in an
elevated position with few scatterers around, and the propagation
between the BS and the user is mainly characterized by rich local
scatterers around the user (see Fig. \ref{fig6}). Assuming the
propagation consists of $P$ i.i.d. paths, we have
\begin{align}
\boldsymbol{h}=\frac{1}{\sqrt{P}}\sum_{p=1}^P\alpha_{p}\boldsymbol{a}(\theta_{p})
\label{channel-model}
\end{align}
where $\alpha_{p}\sim \mathcal{CN}(0,\xi^2)$ denotes the fading
coefficient associated with the $p$th path, and
$\boldsymbol{a}(\theta_{p})$ is the steering vector. For a uniform
linear array, it is given as
\begin{align}
\boldsymbol{a}(\theta_{p}) \triangleq
[1\phantom{0}e^{-j(2\pi/\chi)d\text{cos}(\theta_{p})}
\phantom{0}\ldots \phantom{0}
e^{-j(M-1)(2\pi/\chi)d\text{cos}(\theta_{p})}]^T
\end{align}
in which $\chi$ is the signal wavelength, $d$ denotes the distance
between neighboring antenna elements, and $\theta_p\in [0,\pi]$ is
the azimuth angle of arrival (AoA) of the $p$th path. In the
one-ring mode, the user is surrounded by rich local scatterers
with a radius $r$ that is relatively small compared to the
distance between the BS and the user, $D$. Thus the angular spread
at the BS, approximately given as $\delta=\text{arctan}(r/D)$, is
small.




To estimate the channel from the received signal $\boldsymbol{y}$
(c.f. (\ref{data-model})), it is usually required that the number
of pilot symbols (in time), $T$, is no less than the number of
transmitted antennas $M$, i.e. $T\geq M$. When $M$ is large, the
overhead for downlink training and uplink channel state
information (CSI) feedback becomes prohibitively high. Hopefully,
due to the narrow angular spread at the BS, the steering vectors
$\{\boldsymbol{a}(\theta_{p})\}$ of these $P$ paths are highly
correlated, and thus the channel covariance matrix
$\boldsymbol{R}=E[\boldsymbol{h}\boldsymbol{h}^H]$ has an
approximate low-rank structure. This low-rank structure can be
utilized to reduce the overhead for downlink training for FDD
systems, see, e.g. \cite{AdhikaryNam13,ChoiLove14,NohZoltowski14}.



In this paper, we study the behavior of the minimum mean-squared
error (MMSE) estimator when the the channel covariance matrix has
a low rank structure. We conduct a quantitative analysis to
investigate how much training overhead reduction can be achieved
by exploiting the low-rank structure of the channel covariance
matrix. Assume $\boldsymbol{h}$ is zero-mean complex Gaussian with
covariance matrix $\boldsymbol{R}$, the MMSE estimate of the
channel $\boldsymbol{h}$ is given as
\begin{align}
\boldsymbol{\hat{h}}=\boldsymbol{R}\boldsymbol{X}^H(\boldsymbol{X}\boldsymbol{R}\boldsymbol{X}^H
+\sigma^2\boldsymbol{I})^{-1}\boldsymbol{y}
\end{align}
Note that the MMSE estimator, with the aid of the statistical
information of the channel, does not require an invertible pilot
matrix $\boldsymbol{X}$ (i.e. $T\geq M$) for channel estimation.
The mean-squared error (MSE) associated with the MMSE estimate is
given by
\begin{align}
\text{MSE}=&E\left[\|\boldsymbol{\hat{h}}-\boldsymbol{h}\|_2^2\right]
\nonumber\\
=&\text{tr}\left(\boldsymbol{R}-\boldsymbol{R}\boldsymbol{X}^H(\boldsymbol{X}\boldsymbol{R}\boldsymbol{X}^H
+\sigma^2\boldsymbol{I})^{-1}\boldsymbol{X}\boldsymbol{R}\right)
\label{h-mse}
\end{align}

\section{Asymptotic Behavior of The MMSE} \label{sec:low-noise-MMSE} In this section, we first study the
behavior of the MMSE estimator in the asymptotic low-noise regime,
i.e. $\sigma^2\rightarrow 0$. Our asymptotic analysis shows that a
perfect channel recovery from a small number of symbols is
possible when the channel covariance matrix has a low-rank
structure. Our main results are summarized as follows.


\newtheorem{theorem}{Theorem}
\begin{theorem}
Consider the channel estimation problem described in
(\ref{data-model}), where
$\boldsymbol{h}\sim\mathcal{N}(\boldsymbol{0},\boldsymbol{R})$ and
the rank of the channel covariance matrix $\boldsymbol{R}$ is
$r=\text{rank}(\boldsymbol{R})$. Define
$\boldsymbol{\Phi}\triangleq\boldsymbol{R}^{\frac{1}{2}}\boldsymbol{X}^{H}\boldsymbol{X}\boldsymbol{R}^{\frac{1}{2}}$.
Let
$\boldsymbol{\Phi}=\boldsymbol{V}\boldsymbol{\Gamma}\boldsymbol{V}^{H}$
denote the eigenvalue decomposition (EVD) of $\boldsymbol{\Phi}$,
where $\boldsymbol{V}\triangleq
[\boldsymbol{v}_1\phantom{0}\ldots\phantom{0}\boldsymbol{v}_M]$ is
a unitary matrix consisting of eigenvectors of
$\boldsymbol{\Phi}$, and
$\boldsymbol{\Gamma}=\text{diag}(\gamma_{1},\ldots,\gamma_{r},0,\ldots,0)$
is a diagonal matrix with
$\gamma_{1}\geq\gamma_{2}\geq\ldots\geq\gamma_{r}> 0$. Suppose the
pilot signal $\boldsymbol{X}$ is randomly generated, and the
number of symbols, $T$, is no less than $r$, i.e. $T\geq r$, then
the MSE of the MMSE estimate of $\boldsymbol{h}$ is given by
\begin{align}
E\left[\|\boldsymbol{\hat{h}}-\boldsymbol{h}\|_2^2\right]
=\sum_{i=1}^r\left(1+\gamma_i/\sigma^2\right)^{-1}\boldsymbol{v}_i^H\boldsymbol{R}\boldsymbol{v}_i
\end{align}
and the MSE approaches zero in the limit of vanishing noise, that
is,
\begin{align}
\lim_{\sigma^2\rightarrow
0}E\left[\|\boldsymbol{\hat{h}}-\boldsymbol{h}\|_2^2\right]=0
\nonumber
\end{align}
\label{theorem1}
\end{theorem}
\begin{proof}
Using the Woodbury identity, the MSE (\ref{h-mse}) can be
rewritten as
\begin{align}
&E\left[\|\boldsymbol{\hat{h}}-\boldsymbol{h}\|_2^2\right] \nonumber\\
=&\text{tr}\left(\boldsymbol{R}^{1/2}(\boldsymbol{I}-\boldsymbol{R}^{1/2}
\boldsymbol{X}^H(\boldsymbol{X}\boldsymbol{R}\boldsymbol{X}^H+\sigma^2\boldsymbol{I})^{-1}
\boldsymbol{X}\boldsymbol{R}^{1/2})\boldsymbol{R}^{1/2}\right)
\nonumber\\
=&\text{tr}\left(\boldsymbol{R}^{1/2}(\boldsymbol{I}+\sigma^{-2}\boldsymbol{R}^{1/2}\boldsymbol{X}^H
\boldsymbol{X}\boldsymbol{R}^{1/2})^{-1}
\boldsymbol{R}^{1/2}\right) \nonumber\\
=&
\text{tr}\left(\boldsymbol{R}\boldsymbol{V}(\sigma^{-2}\boldsymbol{\Gamma}+
\boldsymbol{I})^{-1}\boldsymbol{V}^H \right) \nonumber\\
=&\sum_{i=1}^r\left(1+\gamma_i/\sigma^2\right)^{-1}\boldsymbol{v}_i^H\boldsymbol{R}\boldsymbol{v}_i
+\sum_{i=r+1}^M \boldsymbol{v}_i^H\boldsymbol{R}\boldsymbol{v}_i
\label{mse-decomposition}
\end{align}
We can see that the first term in (\ref{mse-decomposition})
vanishes as $\sigma^2\rightarrow 0$.


We now examine under what conditions the second term in
(\ref{mse-decomposition}) reduces to zero. Let
\begin{align}
\boldsymbol{R}=\boldsymbol{U}\boldsymbol{\Lambda}\boldsymbol{U}^H
\end{align}
denote the reduced EVD of $\boldsymbol{R}$, where
$\boldsymbol{U}\in\mathbb{C}^{M\times r}$ and
$\boldsymbol{\Lambda}\in\mathbb{C}^{r\times r}$. We can write
\begin{align}
\boldsymbol{R}^{\frac{1}{2}}\boldsymbol{X}^{H}=\boldsymbol{U}\boldsymbol{\Lambda}^{\frac{1}{2}}
\boldsymbol{U}^H\boldsymbol{X}^{H}=\boldsymbol{U}\boldsymbol{C}
\end{align}
where $\boldsymbol{C}\triangleq\boldsymbol{\Lambda}^{\frac{1}{2}}
\boldsymbol{U}^H\boldsymbol{X}^{H}\in\mathbb{C}^{r\times T}$. When
$T\geq r$ and the pilot symbols of $\boldsymbol{X}$ are randomly
generated according to some distribution, the matrix
$\boldsymbol{C}$ has a full row rank with probability one, i.e.
$\text{rank}(\boldsymbol{C})=r$. Thus we have
\begin{align}
\text{Range}(\boldsymbol{R}^{\frac{1}{2}}\boldsymbol{X}^{H})=\text{Range}(\boldsymbol{\Phi})
=\text{Range}(\boldsymbol{U}) \label{range-space}
\end{align}
where $\text{Range}(\boldsymbol{A})$ denotes the column space
spanned by the column vectors of $\boldsymbol{A}$. From
(\ref{range-space}), we can immediately arrive at
\begin{align}
\boldsymbol{u}_i^H\boldsymbol{R}=\boldsymbol{v}_i^H\boldsymbol{U}=\boldsymbol{v}_i^H\boldsymbol{\Phi}
=\boldsymbol{0}
\quad \forall i=r+1,\ldots,M
\end{align}
Hence the second term in (\ref{mse-decomposition}) disappears
provided that the length of the pilot in time is no less than the
rank of the channel covariance matrix, i.e. $T\geq r$, and
eventually we reach the conclusion that the MSE of the MMSE
estimate of $\boldsymbol{h}$ approaches zero in the limit of
vanishing noise, that is,
\begin{align}
\lim_{\sigma^2\rightarrow
0}E\left[\|\boldsymbol{\hat{h}}-\boldsymbol{h}\|_2^2\right]=0
\nonumber
\end{align}
The proof is completed here.
\end{proof}



\emph{Discussions:} The significance of Theorem \ref{theorem1}
lies in that, in the limit of vanishing noise, it establishes
sufficient conditions for the MMSE estimator to achieve exact
channel recovery from only a small number of pilot symbols. We
note that another line of research \cite{RaoLau14,GaoDai15} for
FDD downlink training and channel estimation exploits the sparsity
of the channel on the virtual angular domain and formulates the
channel estimation as a compressed sensing problem:
\begin{align}
\boldsymbol{y}=\boldsymbol{X}\boldsymbol{h}+\boldsymbol{w}=\boldsymbol{X}\boldsymbol{A}\boldsymbol{\tilde{h}}
+\boldsymbol{w}
\end{align}
where $\boldsymbol{A}$ is a basis for the virtual angular domain.
For the uniform linear array case, the basis $\boldsymbol{A}$ is a
discrete Fourier transform (DFT) matrix. $\boldsymbol{\tilde{h}}$
is a sparse vector to be estimated. This class of approaches are
justified by compressed sensing theories, which assert that a
sparse signal can be perfectly recovered from compressive
measurements, provided that the measurement matrix satisfies a
certain RIP condition \cite{CandesTao05}. Our theorem here can be
regarded as a counterpart result for the MMSE estimator, and
provides a justification for using the MMSE estimator for channel
estimation from a small number of pilot symbols.

It is also interesting to compare conditions required by the MMSE
estimator and those by compressed sensing techniques to achieve
perfect channel recovery. First recall the following lemma that
characterizes the number of dimensions of a subspace spanned by a
number of steering vectors with a bounded support of angles of
arrival (AoAs):

\newtheorem{lemma}{Lemma}
\begin{lemma} \label{lemma1}
Define
\begin{align}
\boldsymbol{\alpha}(x)\triangleq [1\phantom{0}e^{-j\pi
x}\phantom{0}\ldots\phantom{0}e^{-j\pi (M-1)x}]^T
\label{DFT-basis}
\end{align}
and $\mathcal{A}\triangleq\text{span}\{\boldsymbol{\alpha}(x),x\in
[-1,1)\}$. Given $b_1,b_2\in [-1,1]$ and $b_1<b_2$, define
$\mathcal{B}\triangleq\text{span}\{\boldsymbol{\alpha}(x),x\in
[b_1,b_2]\}$, then
\begin{align}
&\text{dim}(\mathcal{A})=M \nonumber\\
&\text{dim}(\mathcal{B})\sim \text{$(b_2-b_1)M/2$ when $M$ grows
large}
\end{align}
\end{lemma}
\begin{proof}
See \cite[Lemma 1]{YinGesbert13}.
\end{proof}

Consider the one-ring model (\ref{channel-model}) with the
multipath angle of arrival $\theta$ distributed on a bounded
support, i.e. $\theta\in
[\theta_{\text{min}},\theta_{\text{max}}]$. From Lemma
\ref{lemma1}, the rank of the channel covariance matrix
$\boldsymbol{R}$ is upper bounded by
\begin{align}
\text{rank}(\boldsymbol{R})\leq \eta M \quad \text{as
$M\rightarrow\infty$}
\end{align}
where $\eta$ is defined as
\begin{align}
\eta\triangleq
|\cos(\theta_{\text{min}})-\cos(\theta_{\text{max}})|d/\chi
\end{align}
in which $d$ denotes the distance between neighboring antennas and
$\chi$ is the signal wavelength. Another important property from
Lemma \ref{lemma1} is that, when $M\rightarrow\infty$, the channel
$\boldsymbol{h}$ has a sparse representation on a virtual angular
domain with $r=\text{rank}(\boldsymbol{R})$ nonzero coefficients.

With the above results, we are now ready to make a fair comparison
between conditions required by the MMSE estimator and,
respectively, by the compressed sensing methods for exact recovery
of the channel. For the MMSE estimator, from Theorem
\ref{theorem1}, we know that as few as $T=r$ symbols are needed to
perfectly recover the channel. On the other hand, for compressed
sensing-based methods, it has been shown that the number of
required measurements for exact recovery is of order
$T=\mathcal{O}(r\log(M/r))$ using polynomial-time optimization
solvers or greedy algorithms \cite{CandesTao05}. If the
computational complexity is not a concern, then at least $T=2r$
measurements are required for exact recovery via the
$\ell_0$-minimization. From the above discussion, we can see that
the MMSE estimator requires fewer symbols than compressed sensing
techniques for exact channel recovery. This result puts the
covariance-aided methods into a favorable position for FDD
downlink training and channel estimation.




\section{Optimal Pilot Sequence Design} \label{sec:pilot-design-single-user}
Our analysis in the previous section reveals that as few as $T=r$
symbols in time are required to guarantee perfect channel recovery
in the asymptotic low-noise regime, i.e. $\sigma^2\rightarrow 0$.
Nevertheless, assuming a noiseless scenario is unrealistic in
practical systems. Therefore it is meaningful to study the
behavior of the MMSE estimator for a non-vanishing $\sigma^2$. For
the case $\sigma^2\neq 0$, we would like to examine whether a
larger value of $T$ leads to a better estimation accuracy, or if
$T=r$ is sufficient to attain a minimum MSE. To answer this
question, we first need to impose a power constraint on the pilot
signal, i.e. $\text{tr}(\boldsymbol{X}\boldsymbol{X}^H)\leq P$;
otherwise a fair comparison between pilots of different lengths is
impossible. Note that different pilots of the same length also
result in different MSEs. Hence simply comparing the MSEs attained
by two arbitrary pilots of different lengths does not provide any
meaningful answers. To make sense, we have to compare the MSEs
attained by optimally devised pilots for different values of $T$,
and see if increasing $T$ will result in a lower MSE. This
requires us to examine the following optimization problem
\begin{align}
\min_{\boldsymbol{X}} & \quad
\text{MSE}=\text{tr}\left(\boldsymbol{R}-\boldsymbol{R}\boldsymbol{X}^H(\boldsymbol{X}\boldsymbol{R}\boldsymbol{X}^H
+\sigma^2\boldsymbol{I})^{-1}\boldsymbol{X}\boldsymbol{R}\right)
\nonumber\\
\text{s.t.} &\quad \text{tr}(\boldsymbol{X}\boldsymbol{X}^H)\leq P
\label{opt-1}
\end{align}
The solution of the above optimization problem is summarized as
follows.

\begin{theorem} \label{theorem2}
Let
$\boldsymbol{R}=\boldsymbol{U}_0\boldsymbol{\Lambda}_0\boldsymbol{U}_0^H$
denote the EVD\footnote{Here
$\boldsymbol{R}=\boldsymbol{U}_0\boldsymbol{\Lambda}_0\boldsymbol{U}_0^H$
is used to distinguish itself from the truncated EVD
$\boldsymbol{R}=\boldsymbol{U}\boldsymbol{\Lambda}\boldsymbol{U}^H$.}
of $\boldsymbol{R}$, where
$\boldsymbol{\Lambda}_0=\text{diag}(\lambda_{1},\ldots,\lambda_{M})$
is a diagonal matrix with its diagonal entries arranged in a
decreasing order and $\boldsymbol{U}_0\in \mathbb{C}^{M\times M}$
is a unitary matrix. The optimal solution to (\ref{opt-1}) is then
given by
\begin{align}
\boldsymbol{X}=[\boldsymbol{\Delta}\phantom{0}
\boldsymbol{0}]\boldsymbol{U}_0^H \label{optimalX}
\end{align}
where $\boldsymbol{\Delta}=\text{diag}(\delta_1,\ldots,\delta_T)$
with $\delta_i$ given as
\begin{align}
\delta_i=
\begin{cases}
\sqrt{\mu-\sigma^2\lambda_{i}^{-1}}& \text{if $\mu\geq\sigma^2\lambda_{i}^{-1}$ and $\lambda_{i}\neq 0$}\\
0& \text{otherwise}
\end{cases} \label{delta}
\end{align}
in which $\mu$ is determined by the constraint
$\sum_{i=1}^T\delta_i^2=P$.
\end{theorem}
\begin{proof}
According to \cite[Theorem 1]{PalomarCioffi03}, the optimal
$\boldsymbol{X}$ has a form of
\begin{align}
\boldsymbol{X}^H=\boldsymbol{U}_0[:,1:T]\boldsymbol{\Delta}
\label{eqn1}
\end{align}
where $\boldsymbol{U}_0[:,1:T]$ consists of $T$ eigenvectors of
$\boldsymbol{R}$ associated with the first $T$ largest
eigenvalues, and
$\boldsymbol{\Delta}=\text{diag}(\delta_1,\ldots,\delta_T)$ is a
diagonal matrix with its diagonal elements to be determined as
follows. Substituting (\ref{eqn1}) into (\ref{opt-1}), the
optimization (\ref{opt-1}) can be simplified as
\begin{align}
\min_{\{\delta_i\}}&\quad
\sum_{i=1}^T\frac{\sigma^2\lambda_{i}}{\delta_i^2\lambda_{i}+\sigma^2}\nonumber\\
\text{s.t.} &\quad \sum_{i=1}^T\delta_i^2\leq P \nonumber\\
&\quad \delta_i^2\geq 0 \quad \forall i=1,\ldots,T \label{opt-2}
\end{align}
The above optimization can be solved analytically by resorting to
the Lagrangian function and Karush-Kuhn-Tucker (KKT) conditions,
which leads to a water-filling type power allocation scheme
described by
\begin{align}
\delta_i=
\begin{cases}
\sqrt{\mu-\sigma^2\lambda_{i}^{-1}}& \text{if $\mu\geq\sigma^2\lambda_{i}^{-1}$ and $\lambda_{i}\neq 0$}\\
0& \text{otherwise}
\end{cases}
\end{align}
where $\mu$ is determined to ensure that the KKT condition
$\sum_{i=1}^T\delta_i^2=P$ is satisfied. The proof is completed
here.
\end{proof}

We now discuss whether a larger value of $T$ would result in a
lower MSE. Note that the MSE achieved by the optimal
$\boldsymbol{X}$ is given by
\begin{align}
\text{MSE}(T)=&\sum_{i=1}^r\lambda_{i}-\sum_{i=1}^{T}
\frac{\delta_{i}^2(T)\lambda_{i}^2}{\delta_{i}^2(T)\lambda_{i}+\sigma^2}
\end{align}
where we use $\delta_i(T)$ to denote the dependence of $\delta_i$
on $T$. The $r$-rank channel covariance matrix $\boldsymbol{R}$
implies $\lambda_{i}=0,\forall i>r$. Considering the case $T>r$,
from (\ref{delta}), it is easy to verify that for any $T>r$, we
have
\begin{align}
\delta_{i}(T)=\begin{cases} \delta_i(r) & \forall i=1,\ldots,r \\
0 & \forall i=r+1,\ldots,T \end{cases}
\end{align}
Therefore we can arrive at
\begin{align}
\text{MSE}(T)=\text{MSE}(r) \quad \forall T>r
\end{align}
On the other hand, from the optimality of the solution
(\ref{delta}), it is clear that
\begin{align}
\text{MSE}(T)\leq\text{MSE}(r) \quad \forall T<r
\end{align}
Based on the above results, we know that the minimum MSE can be
attained by simply choosing $T=r$, and a larger $T$ beyond the
value of $r$ does not lead to a smaller MSE. This result provides
an affirmative answer to the question discussed at the beginning
of this section, that is, given a transmit power constraint
$\text{tr}(\boldsymbol{X}\boldsymbol{X}^H)=P$, a minimum MSE can
be achieved by setting the number of symbols equal to the rank of
the channel covariance matrix, i.e. $T=r$.


Note that the pilot constraint considered here is different from
that of \cite{ChoiLove14}, where unitary training with equal power
allocation per pilot symbol is assumed, i.e.
$\boldsymbol{X}\boldsymbol{X}^H=\rho\boldsymbol{I}$. Clearly, in
this case, the total amount of transmit power increases unbounded
as $T$ becomes large, more precisely, we have
$\text{tr}(\boldsymbol{X}\boldsymbol{X}^H)=\rho T$. Hence a larger
$T$ always leads to an improved channel estimation accuracy.
Nevertheless, for a power-constrained wireless network where
energy efficiency is of a major concern, the power constraint
considered in this paper may be more meaningful.



\section{Asymptotically Optimal Pilot for Multi-User Scenarios} \label{sec:pilot-design-multiuser}
In the previous section, we derived the optimal pilot sequence for
the single-user case. For massive MIMO systems where the BS aims
to simultaneously serve a number of users, the pilot sequence has
to be shared by multiple users. Unfortunately, the channels
associated with these users may not have the same channel
covariance matrix. In this case, it is impossible to find an
optimal pilot sequence $\boldsymbol{X}$ to simultaneously minimize
the MSEs associated with all users. To address this difficulty, in
\cite{JiangMolisch15}, the pilot sequence is designed to maximize
a summation of the conditional mutual information associated with
all users, and an iterative algorithm was developed to solve the
maximization problem. In this section, a different criterion is
considered, where the objective is to minimize the sum of MSEs
associated with all users, i.e.
\begin{align}
\min_{\boldsymbol{X}}\quad &\sum_{k=1}^K \text{MSE}_k \nonumber\\
=& \sum_{k=1}^K
\text{tr}\left(\boldsymbol{R}_k-\boldsymbol{R}_k\boldsymbol{X}^H(\boldsymbol{X}\boldsymbol{R}_k\boldsymbol{X}^H
+\sigma^2\boldsymbol{I})^{-1}\boldsymbol{X}\boldsymbol{R}_k\right)
\nonumber\\
\text{s.t.}\quad &\text{tr}(\boldsymbol{X}\boldsymbol{X}^H)\leq P
\label{opt-3}
\end{align}
where $\boldsymbol{R}_k$ and $\text{MSE}_k$ denote the channel
covariance matrix and the MSE associated with the $k$th user,
respectively. Also, to simplify the problem, we assume the noise
variances across different users are identical, i.e.
$\sigma_1^2=\ldots=\sigma_K^2=\sigma^2$. Finding an analytical
solution to the above optimization is difficult. Nevertheless, we
will show that an asymptotically optimal training sequence can be
devised given that users have mutually non-overlapping angles of
arrival (AoAs). Here the asymptotic optimality means that the
solution approaches the optimal one as the number of antennas at
the BS goes to infinity.


Before proceeding, we first introduce the following properties
which were proved in \cite{YinGesbert13,YouGao15} and reveal the
eigenstructure properties of the channel covariance matrices.
Consider the channel $\boldsymbol{h}$ generated by the one ring
model with a bounded support of angle of arrival $\theta\in
[\theta_{\text{min}},\theta_{\text{max}}]$. Let $\boldsymbol{R}$
denote the channel covariance matrix. We have the following
properties regarding the channel covariance matrices.

\emph{Property 1} \cite[Lemma 3]{YinGesbert13}: In the asymptotic
regime of large number of antennas, steering vectors
$\boldsymbol{a}(\vartheta)$ with $\vartheta \notin
[\theta_{\text{min}},\theta_{\text{max}}]$ fall in the null space
of the covariance matrix $\boldsymbol{R}$, i.e.
\begin{align}
\text{null}(\boldsymbol{R})\supset
\text{span}\{\boldsymbol{a}(\vartheta)/\sqrt{M}, \forall \vartheta
\notin [\theta_{\text{min}},\theta_{\text{max}}]\}, \text{as
$M\rightarrow\infty$}
\end{align}

\emph{Property 2} \cite[Lemma 1]{YouGao15}: For the uniform linear
array (ULA) case, when $M\rightarrow\infty$, the eigenvector
matrix of the channel covariance matrix $\boldsymbol{R}$ can be
well approximated by a unitary discrete Fourier transform (DFT)
matrix.

\emph{Property 3}: From the above two properties, we naturally
arrive at the following property: The column vectors in the DFT
matrix whose angular coordinates are located outside the support
of angle of arrival form an orthonormal basis for the null space
of $\boldsymbol{R}$. Meanwhile, those column vectors in the DFT
matrix whose angular coordinates lie within the support of angle
of arrival form an orthonormal basis for $\boldsymbol{R}$. More
precisely, let
\begin{align}
\boldsymbol{F}\triangleq\frac{1}{\sqrt{M}}
[\boldsymbol{\alpha}(\omega_1)\phantom{0}\boldsymbol{\alpha}(\omega_2)\phantom{0}
\ldots\phantom{0}\boldsymbol{\alpha}(\omega_M)]
\end{align}
denote the DFT matrix, in which $\omega_m=-1+2(m-1)/M,\forall m$,
and $\boldsymbol{\alpha}(\omega_m)$ is defined in
(\ref{DFT-basis}). Let
$\boldsymbol{R}=\boldsymbol{U}\boldsymbol{\Lambda}\boldsymbol{U}^H$
denote the truncated eigenvalue decomposition, where
$\boldsymbol{U}\in\mathbb{C}^{M\times r}$, and
$\boldsymbol{\Lambda}\in\mathbb{C}^{r\times r}$. Then as
$M\rightarrow\infty$, $\boldsymbol{U}$ is composed of column
vectors of $\boldsymbol{F}$ whose angular coordinates
$\{\omega_i\}$ lie within the support of AoA, i.e.
\begin{align}
\boldsymbol{U}=[\boldsymbol{\alpha}(\omega_{i_1})\phantom{0}
\ldots\phantom{0}\boldsymbol{\alpha}(\omega_{i_{r}})]
\end{align}
where $\omega_{i}\in
[2d\cos(\theta_{\text{min}})/\chi,2d\cos(\theta_{\text{max}})/\chi]$
for $i=i_1,\ldots,i_r$.



We now discuss how to devise an asymptotically optimal pilot
sequence for (\ref{opt-3}). Let
\begin{align}
\boldsymbol{R}_k=\boldsymbol{U}_k\boldsymbol{\Lambda}_k\boldsymbol{U}_k^H
\end{align}
denote the truncated eigenvalue decomposition, where
$\boldsymbol{U}_k\in\mathbb{C}^{M\times r_k}$, and
$\boldsymbol{\Lambda}_k\in\mathbb{C}^{r_k\times r_k}$. $r_k$
denotes the rank of $\boldsymbol{R}_k$. For simplicity, we assume
$r_k=r, \forall k$. Inspired by the above properties, we propose
an overlayed pilot sequence that is a superposition of a set of
pilot sequences $\{\boldsymbol{X}_k\}$
\begin{align}
\boldsymbol{X}=\sum_{k=1}^K \boldsymbol{X}_k
\label{optimalX-overlay}
\end{align}
where $\boldsymbol{X}_k$ denotes the pilot sequence optimally
designed for user $k$, i.e. given a power constraint
$\text{tr}(\boldsymbol{X}_k\boldsymbol{X}_k^H)=P_k^{\ast}$,
$\boldsymbol{X}_k$ is given by Theorem \ref{theorem2}, i.e.
\begin{align}
\boldsymbol{X}_k=\boldsymbol{\Delta}_k^H\boldsymbol{U}_k^H
\label{Xk}
\end{align}
in which $\boldsymbol{\Delta}_k$ is a diagonal matrix with its
diagonal elements optimized according to a water-filling power
allocation scheme as described in Theorem \ref{theorem2}.


We now show that the asymptotically optimal solution to
(\ref{opt-3}) has a form of (\ref{optimalX-overlay}). Note that
any pilot sequence $\boldsymbol{X}$ can be expressed in terms of
the DFT matrix as follows
\begin{align}
\boldsymbol{X}=\boldsymbol{Z}\boldsymbol{F}^H
\end{align}
where $\boldsymbol{Z}\in\mathbb{C}^{T\times M}$ is a matrix to be
optimized. Recalling Properties 2 and 3, we have
\begin{align}
\boldsymbol{X}\boldsymbol{R}_k=&\boldsymbol{Z}\boldsymbol{F}^H\boldsymbol{R}_k
\stackrel{(a)}{=}
(\boldsymbol{Z}_k\boldsymbol{U}_k^H+\boldsymbol{\bar{Z}}_k\boldsymbol{\bar{U}}_k^H)\boldsymbol{R}_k
\nonumber\\
=&\boldsymbol{Z}_k\boldsymbol{U}_k^H\boldsymbol{R}_k \label{eqn2}
\end{align}
where $(a)$ comes from the fact that we can partition the DFT
matrix into two parts
$\boldsymbol{F}=[\boldsymbol{U}_k\phantom{0}\boldsymbol{\bar{U}}_k]$,
in which $\boldsymbol{U}_k$ is an orthonormal basis of
$\boldsymbol{R}_k$ and $\boldsymbol{\bar{U}}_k$ is an orthonormal
basis for the null space of $\boldsymbol{R}_k$. Accordingly,
$\boldsymbol{Z}$ can be partitioned into two parts:
$\boldsymbol{Z}=[\boldsymbol{Z}_k\phantom{0}\boldsymbol{\bar{Z}}_k]$,
where $\boldsymbol{Z}_k\in\mathbb{C}^{T\times r}$ is a submatrix
of $\boldsymbol{Z}$ consisting of $r$ column vectors. Substituting
(\ref{eqn2}) into the objective function (\ref{opt-3}), we have
\begin{align}
&\sum_{k=1}^K
\text{tr}\left(\boldsymbol{R}_k-\boldsymbol{R}_k\boldsymbol{X}^H(\boldsymbol{X}\boldsymbol{R}_k\boldsymbol{X}^H
+\sigma^2\boldsymbol{I})^{-1}\boldsymbol{X}\boldsymbol{R}_k\right)
\nonumber\\
=& \sum_{k=1}^K
\text{tr}\left(\boldsymbol{R}_k-\boldsymbol{R}_k\boldsymbol{U}_k\boldsymbol{Z}_k^H
(\boldsymbol{Z}_k\boldsymbol{U}_k^H\boldsymbol{R}_k\boldsymbol{U}_k\boldsymbol{Z}_k^H
+\sigma^2\boldsymbol{I})^{-1}\boldsymbol{Z}_k\boldsymbol{U}_k^H\boldsymbol{R}_k
\right)
\end{align}
Since users have mutually non-overlapping AoAs, each matrix
$\boldsymbol{Z}_k$ is constructed by $r$ unique columns of
$\boldsymbol{Z}$ that are not shared by other matrices
$\boldsymbol{Z}_{\bar{k}},\forall \bar{k}\neq k$. Therefore the
optimization (\ref{opt-3}) can be decomposed into $K$ independent
problems, with $\boldsymbol{Z}_k$ optimized in each individual
problem
\begin{align}
\min_{\boldsymbol{Z}_k}\quad&\text{tr}\left(\boldsymbol{R}_k-\boldsymbol{R}_k\boldsymbol{U}_k\boldsymbol{Z}_k^H
(\boldsymbol{Z}_k\boldsymbol{U}_k^H\boldsymbol{R}_k\boldsymbol{U}_k\boldsymbol{Z}_k^H
+\sigma^2\boldsymbol{I})^{-1}\boldsymbol{Z}_k\boldsymbol{U}_k^H\boldsymbol{R}_k
\right)
\nonumber\\
\text{s.t.}\quad &\text{tr}(\boldsymbol{Z}_k\boldsymbol{Z}_k^H)=
P_k^{\ast}
\end{align}
where $P_k^{\ast}$ is the optimal power allocated to the $k$th
user\footnote{Our objective is to show the asymptotically optimal
pilot sequence has a form of (\ref{optimalX-overlay}). The search
of the optimal power allocation $\{P_k^{\ast}\}$ is not considered
here.}. From Theorem \ref{theorem2}, we know that setting $T=r$ is
sufficient to achieve a minimum MSE and the optimal
$\boldsymbol{Z}_k$ is a diagonal matrix
\begin{align}
\boldsymbol{Z}_k^{\ast}=\boldsymbol{\Delta}_k^H
\end{align}
with its diagonal elements determined according to a water-filling
power allocation scheme (see Theorem \ref{theorem2}) such that the
constraint
$\text{tr}(\boldsymbol{Z}_k\boldsymbol{Z}_k^H)=P_k^{\ast}$ is
satisfied. For those columns of $\boldsymbol{Z}$ that are not
included in $\{\boldsymbol{Z}_k\}_{k=1}^K$, since they make no
difference to the objective function value, they should be set to
zero in order to save the transmit power. Therefore the
asymptotically optimal pilot signal $\boldsymbol{X}$ can be
written as
\begin{align}
\boldsymbol{X}=&\sum_{k=1}^K\boldsymbol{Z}_k^{\ast}\boldsymbol{U}_k^H
=\sum_{k=1}^K\boldsymbol{\Delta}_k^H\boldsymbol{U}_k^H
\nonumber\\
=& \sum_{k=1}^K\boldsymbol{X}_k
\end{align}
which is a superposition of a set of pilot sequences, with each
pilot sequence optimally designed for each individual user.


\emph{Remark 1:} The above overlayed pilot design has an intuitive
explanation. Given that the AoAs of all users are distinct, from
Property 1, we know that the channel of each user is
asymptotically orthogonal to the channel covariance matrices
associated with other users as $M\rightarrow\infty$, i.e.
$\boldsymbol{h}_k^H\boldsymbol{R}_{k'}=\boldsymbol{0},\forall
k\neq k'$. As a result, we have
$\boldsymbol{X}_{k'}\boldsymbol{h}_k=\boldsymbol{0},\forall k\neq
k'$ for the pilot sequence $\{\boldsymbol{X}_{k'}\}$ devised in
(\ref{Xk}). Hence from the user's perspective, only the optimal
pilot signal will be received, while other non-optimal pilot
signals are filtered when propagating through the channel.



\emph{Remark 2:} The proposed overlayed downlink training scheme
bears a resemblance to the joint spatial division and multiplexing
(JSDM) strategy \cite{AdhikaryNam13}, where a prebeaforming matrix
is employed to reduce the dimension of the channel to be
estimated. In particular, the prebeamforming matrix suggested by
\cite{AdhikaryNam13} is a concatenation of
$\{\boldsymbol{U}_k\}_{k=1}^K$. Although both the proposed
overlayed training scheme and the JSDM scheme use the eigenvectors
of the channel covariance matrices for downlink training, the
rationale behind these two schemes are different. The JSDM scheme
is shown to be asymptotically optimal in terms of the achievable
capacity, whereas the asymptotic optimality of the proposed
overlayed training scheme is established from the channel
estimation perspective. Finally, we remark that a coordination
strategy can be used to make sure that users to be served in the
same time-frequency slot are well separated in the AoA domain,
similarly as discussed in \cite{YinGesbert13,AdhikaryNam13}.

\section{Estimated Covariance-Assisted MMSE} \label{sec:estimated-covariance-MMSE}
The MMSE estimator assumes perfect knowledge of the downlink
channel covariance matrix. This knowledge, however, is unavailable
and needs to be estimated in practice. If the covariance matrix is
estimated by the user, it needs to be fed back to the BS through
some control channel, which involves a significant amount of
overhead. One way to overcome this difficulty is to estimate the
downlink channel covariance matrix from the uplink covariance
matrix, e.g. \cite{LiangChin01,HochwaldMarzetta01}. This approach,
however, still requires a certain amount of specific uplink
training. In this section, we develop a simple scheme to estimate
the channel covariance matrix based on the one ring model. A MMSE
estimator is then constructed based on the estimated covariance
matrix. Our simulation results indicate that the covariance
estimation scheme is effective and can obtain notable improvement
in estimation performance.





According to the one-ring model (\ref{channel-model}), the
covariance matrix of $\boldsymbol{h}$ can be written as
\begin{align}
\boldsymbol{R}=\frac{\xi^2}{P}\sum_{i=1}^P
E[\boldsymbol{a}(\theta_p)\boldsymbol{a}(\theta_p)^H]=\xi^2E[\boldsymbol{a}(\theta)\boldsymbol{a}(\theta)^H]
\label{eqn5}
\end{align}
To calculate $E[\boldsymbol{a}(\theta)\boldsymbol{a}(\theta)^H]$,
we need to know the distribution of $\theta$. Here we assume
$\theta$ is uniformly distributed with mean angle $\bar{\theta}$
and angular spread $\nu$. Thus the $(m,n)$th entry of
$\boldsymbol{R}$ can be expressed as
\begin{align}
R_{mn}=\frac{\xi^2}{2\nu}\int_{\bar{\theta}-\nu}
^{\bar{\theta}+\nu}
e^{-j2\pi\frac{(m-n)d}{\chi}\cos({\theta})}d\theta \label{Rmn}
\end{align}
The above integration, however, is difficult to calculate. Noting
that the angular spread $\nu$ is usually small, we can use the
Taylor expansion of $\cos(\theta)$ to approximate the integral. We
have
\begin{align}
\cos(\theta)\approx\cos(\bar{\theta})-\sin(\bar{\theta})(\theta-\bar{\theta})
\label{eqn3}
\end{align}
Substituting (\ref{eqn3}) into (\ref{Rmn}), we arrive at
\begin{align}
R_{mn}\approx&\frac{\xi^2}{2\nu} e^{jA_{mn}\cos(\bar{\theta})}
\int_{-\nu}^{\nu}
e^{-jA_{mn}\sin(\bar{\theta})\theta} d\theta \nonumber\\
=& \xi^2 e^{jA_{mn}\cos(\bar{\theta})}
\text{sinc}(A_{mn}\sin(\bar{\theta})\nu) \label{Rmn-final}
\end{align}
where $A_{mn}\triangleq2\pi(m-n)d/\chi$, and
$\textrm{sinc}(x)\triangleq\sin(x)/x$ is the sinc function.
Therefore, the covariance matrix $\boldsymbol{R}$ can be
approximated as a parametric matrix with parameters $\bar{\theta}$
and $\nu$. Note that the parameter $\xi^2$ in (\ref{Rmn-final})
can be ignored since as a scaling factor, it is independent of the
signal subspace of $\boldsymbol{R}$. Thus the channel covariance
estimation problem is simplified to find the mean angle
$\bar{\theta}$ and the angular spread $\nu$. There are several
ways to estimate these two parameters. Here we introduce a
compressed sensing-based method. Recalling that the channel with a
narrow angular spread has an approximate sparse representation on
the angular domain, i.e.
\begin{align}
\boldsymbol{h}=\boldsymbol{A}\boldsymbol{\tilde{h}}
\end{align}
where $\boldsymbol{A}$ is an $M\times M$ unitary matrix determined
by the array geometry at the base station. For the uniform linear
array, $\boldsymbol{A}$ becomes the DFT matrix consisting of
columns characterized by different angular coordinates.
$\boldsymbol{\tilde{h}}$ is an approximately sparse vector, of
which the $m$th element is contributed by the paths around the
$m$th angular coordinate. Due to the narrow angular spread, a
majority of the channel energy is concentrated on a few
consecutive angular coordinates. Hence the mean angle and angular
spread can be coarsely estimated from the sparse signal
$\boldsymbol{\tilde{h}}$. More precisely, the angular coordinate
which has the largest magnitude can be estimated as the mean
angle, i.e.
\begin{align}
\hat{\bar{\theta}}=
\begin{cases}
\arccos\left[\frac{\chi}{d}\left(\frac{s-1}{M}\right)\right]& \text{if $s\leq\frac{M}{2}+1$ }\\
\arccos\left[\frac{\chi}{d}\left(\frac{s-1}{M}-1\right)\right]&
\text{otherwise}
\end{cases}
\end{align}
where $s$ is the index of the angular coordinate which has the
largest magnitude, i.e. the $s$th element of
$\boldsymbol{\tilde{h}}$ has the largest magnitude. The angular
spread can be estimated as a symmetric interval around the mean
angle, say,
$[\hat{\bar{\theta}}-\hat{\nu},\hat{\bar{\theta}}+\hat{\nu}]$,
with a majority of the channel energy (say, $90\%$) included in
this interval. Now it remains to estimate the sparse vector
$\boldsymbol{\tilde{h}}$. As indicated earlier in this paper, the
estimation of $\boldsymbol{\tilde{h}}$ can be formulated into a
sparse signal recovery problem:
\begin{align}
\boldsymbol{y}=\boldsymbol{X}\boldsymbol{h}+\boldsymbol{w}=\boldsymbol{X}\boldsymbol{A}\boldsymbol{\tilde{h}}
+\boldsymbol{w} \label{eqn4}
\end{align}
and can be efficiently solved via greedy or convex optimization
methods. After $\boldsymbol{\tilde{h}}$ is recovered, the mean
angle and the angular spread can be obtained by using the
aforementioned procedure, and an estimate of the channel
covariance matrix can be computed by substituting the estimated
mean angle and angular spread into (\ref{Rmn-final}). Finally, a
MMSE estimate of $\boldsymbol{h}$ can be obtained.

For clarity, we summarize our proposed estimated
covariance-assisted MMSE scheme in Algorithm 1.

\begin{algorithm}
\caption{Estimated Covariance-Assisted MMSE (EC-MMSE)}
\label{alg:A}
\begin{algorithmic}
\STATE { Given the received signal $\boldsymbol{y}\in\mathbb{C}^T$
and the pilot signal $\boldsymbol{X}\in\mathbb{C}^{T\times M}$.
\begin{description}
  \item[1] Recover $\boldsymbol{\tilde{h}}$ from
  $\boldsymbol{y}=\boldsymbol{X}\boldsymbol{A}\boldsymbol{\tilde{h}}+\boldsymbol{w}$
      via compressed sensing techniques, where $\boldsymbol{A}$
      is a DFT matrix for the uniform linear array case.
  \item[2] Estimate the mean angle $\hat{\bar{\theta}}$ and angular spread
  $\hat{\nu}$ based on $\boldsymbol{\tilde{h}}$, then obtain an estimate of the channel
  covariance matrix, $\boldsymbol{\hat{R}}$, via (\ref{Rmn-final}).
  \item[3]  Construct a MMSE estimator
  $\boldsymbol{\hat{h}}=\boldsymbol{\hat{R}}\boldsymbol{X}^H(\boldsymbol{X}\boldsymbol{\hat{R}}\boldsymbol{X}^H
      +\sigma^2\boldsymbol{I})^{-1}\boldsymbol{y}$ to estimate the channel $\hat{\boldsymbol{h}}$.
\end{description}
}
\end{algorithmic}
\end{algorithm}

\emph{Remark 1:} Although $\boldsymbol{h}$ can be directly
estimated from (\ref{eqn4}) via compressed sensing techniques, the
MMSE estimator with the help of the estimated channel covariance
matrix can provide a better estimation accuracy, as demonstrated
by our simulation results. Our proposed MMSE estimator can be
employed either at the mobile station (i.e. user) or at the BS to
estimate the channel. If the channel is estimated by the mobile
station, the full CSI needs to be fed back to the BS, which causes
a large amount of uplink overhead when $M$ is large. An
alternative approach is to let the mobile station simply feed back
the received signal $\boldsymbol{y}$ to the BS, and let the BS
form an estimate of the channel based on $\boldsymbol{y}$. This
approach requires less uplink overhead since the dimension of
$\boldsymbol{y}$ is usually smaller than the dimension of the
channel $\boldsymbol{h}$. It should be noted for our proposed
method, the received pilot signal $\boldsymbol{y}$ is used for
both the channel covariance matrix and the channel estimation.
Thus no additional overhead is required.

\emph{Remark 2:} Our scheme assumes a uniform angle of arrival
(AoA) distribution when estimating the channel covariance matrix.
In practice, the AoA may not strictly follow a uniform
distribution. Nevertheless, note that the eigenvectors of the
channel covariance matrix are more closely related to the location
of the interval over which the AoA is distributed, but less
dependent on the specific distribution of the AoA. Hence our
estimation scheme which assumes a uniform AoA distribution can
still reliably estimate the true dominant eigenvectors when there
is a mismatch between the presumed AoA distribution and the true
distribution. As a result, the proposed MMSE estimator still
deliver superior performance, as verified by our simulation
results.

\emph{Remark 3:} The estimation of the channel covariance matrix
based on the one-ring model was also considered in
\cite{NohZoltowski14}. Specifically, the work
\cite{NohZoltowski14} suggested to estimate the channel covariance
matrix as
$\boldsymbol{\hat{R}}=\boldsymbol{F}\boldsymbol{D}\boldsymbol{F}^H$,
where $\boldsymbol{F}$ is a DFT matrix and $\boldsymbol{D}$ is a
diagonal matrix that contains the angular power spectral values.
Our simulation results, however, show that, for a finite number of
antennas, this covariance estimation approximation is not accurate
enough and a MMSE estimator based on this covariance approximation
even leads to deteriorated estimation performance. In addition, as
indicated in \cite{NohZoltowski14}, the estimation of the angular
power spectrum requires additional training overhead and
computational cost.

\section{Simulation Results} \label{sec:simulation}
We now carry out experiments to validate our theoretical results
and to illustrate the performance of the estimated
covariance-assisted MMSE estimator (referred to as EC-MMSE)
proposed in Section \ref{sec:estimated-covariance-MMSE}.
Throughout our simulations, unless otherwise explicitly specified,
we assume a uniform linear array with $M=64$ antennas, and the
distance between neighboring antenna elements is set to a half of
the wavelength of the signal.

\begin{figure*}[!t]
 \subfigure[Exact low-rank covariance matrix]{\includegraphics[width=3.5in]{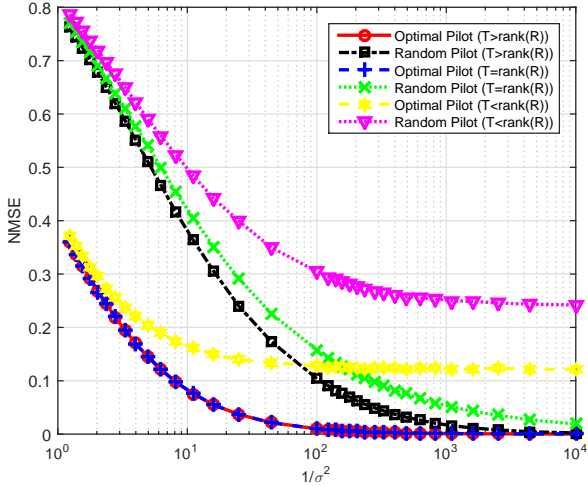}} \hfil
\subfigure[Approximate low-rank covariance
matrix]{\includegraphics[width=3.5in]{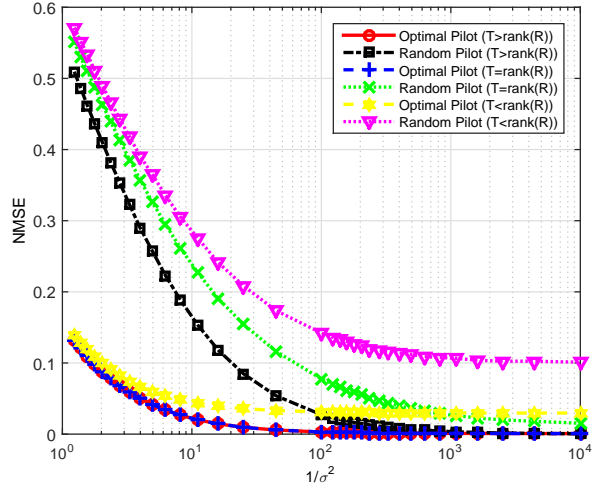}} \caption{NMSE
versus $1/\sigma^2$ for different choices of $T$.} \label{fig1}
\end{figure*}

We first examine the behavior of the MMSE estimator in the
asymptotic low-noise regime when the channel covariance matrix has
a low-rank or an approximate low-rank structure. The channel
covariance matrix is assumed perfectly known by the MMSE
estimator. Fig. \ref{fig1} depicts the normalized mean-squared
errors (NMSEs) of the MMSE estimator vs. the reciprocal of the
noise variance, where we consider both the optimal pilot sequence
devised according to Theorem \ref{theorem2} and a random pilot
sequence whose entries are i.i.d. normal random variables. Note
that the random pilot sequence has to be multiplied by a scaling
factor to satisfy a power constraint
$\text{tr}(\boldsymbol{X}\boldsymbol{X}^H)\leq P$ that is also
imposed on the optimal pilot. In Fig. \ref{fig1}(a), we randomly
generate an exact low-rank channel covariance matrix
$\boldsymbol{R}$ whose rank is set equal to 15. While for Fig.
\ref{fig1}(b), the channel covariance matrix is generated
according to the one-ring model, where the AoAs are assumed to be
uniformly distributed over an interval
$[\bar{\theta}-\nu,\bar{\theta}+\nu]$, with the mean angle and the
angular spread given respectively by $\bar{\theta}=\pi/6$ and
$\nu=\pi/10$, the total number of i.i.d. paths is set to $P=100$,
and $\alpha_p$ follows a complex Gaussian distribution with zero
mean and variance $\xi^2=1$. A numerical average is utilized to
compute (\ref{eqn5}) and obtain the channel covariance matrix for
the one-ring model. Numerical results show that the covariance
matrix has an approximate low-rank structure with about $12$
dominant eigenvalues. To examine the impact of the number of pilot
symbols on the estimation performance, we consider three different
choices of $T$ in our simulations, namely,
$T=20>\text{rank}(\boldsymbol{R})$,
$T=\text{rank}(\boldsymbol{R})$, and
$T=10<\text{rank}(\boldsymbol{R})$. From Fig. \ref{fig1}, we
observe that when the number of symbols $T$ is no less than the
rank of the channel covariance matrix, the NMSE of the MMSE
estimator approaches zero in the limit of vanishing noise, i.e.
$\sigma^2\rightarrow 0$, whatever an optimal pilot sequence or a
random pilot sequence is employed. On the other hand, when
$T<\text{rank}(\boldsymbol{R})$, there exists an error floor for
both the optimal and random pilots, that is, once the error floor
is reached, a decrease in the noise power does not bring any
additional estimation performance improvement. This result
corroborates our theoretical analysis in Section
\ref{sec:low-noise-MMSE}. Also, given a power constraint, the
optimal pilot sequences for $T>\text{rank}(\boldsymbol{R})$ and
$T=\text{rank}(\boldsymbol{R})$ are identical. Thus the NMSEs
achieved by optimal pilot sequences remain unaltered for these two
cases.

\begin{figure*}[!t]
 \subfigure[NMSE vs. SNR]{\includegraphics[width=3.5in]{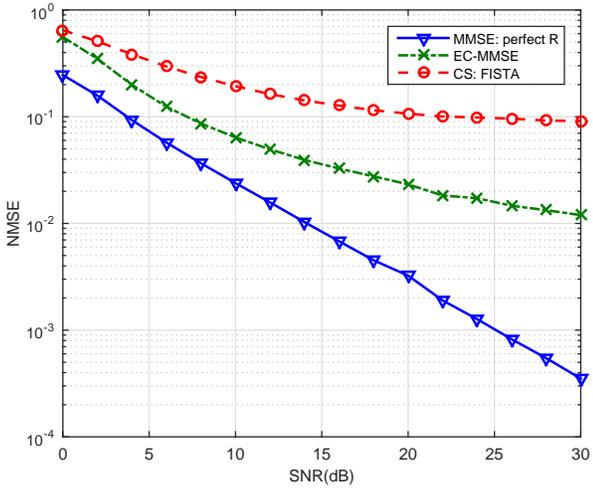}} \hfil
\subfigure[NMSE vs. number of symbols
$T$]{\includegraphics[width=3.5in]{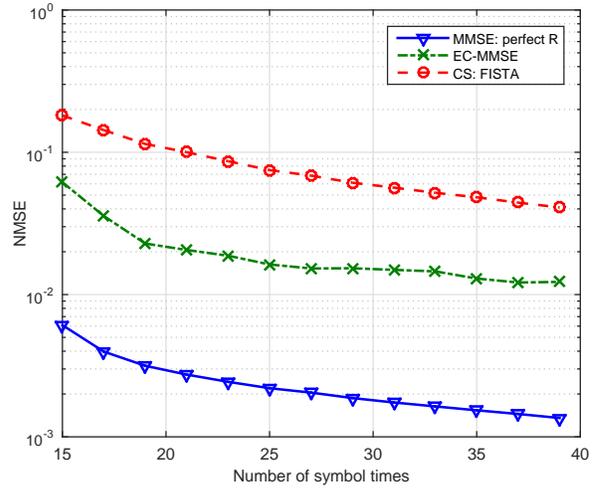}} \caption{NMSEs of
respective schemes vs. SNR and number of symbols $T$} \label{fig2}
\end{figure*}

\begin{figure*}[!t]
 \subfigure[EC-MMSE]{\includegraphics[width=3.5in]{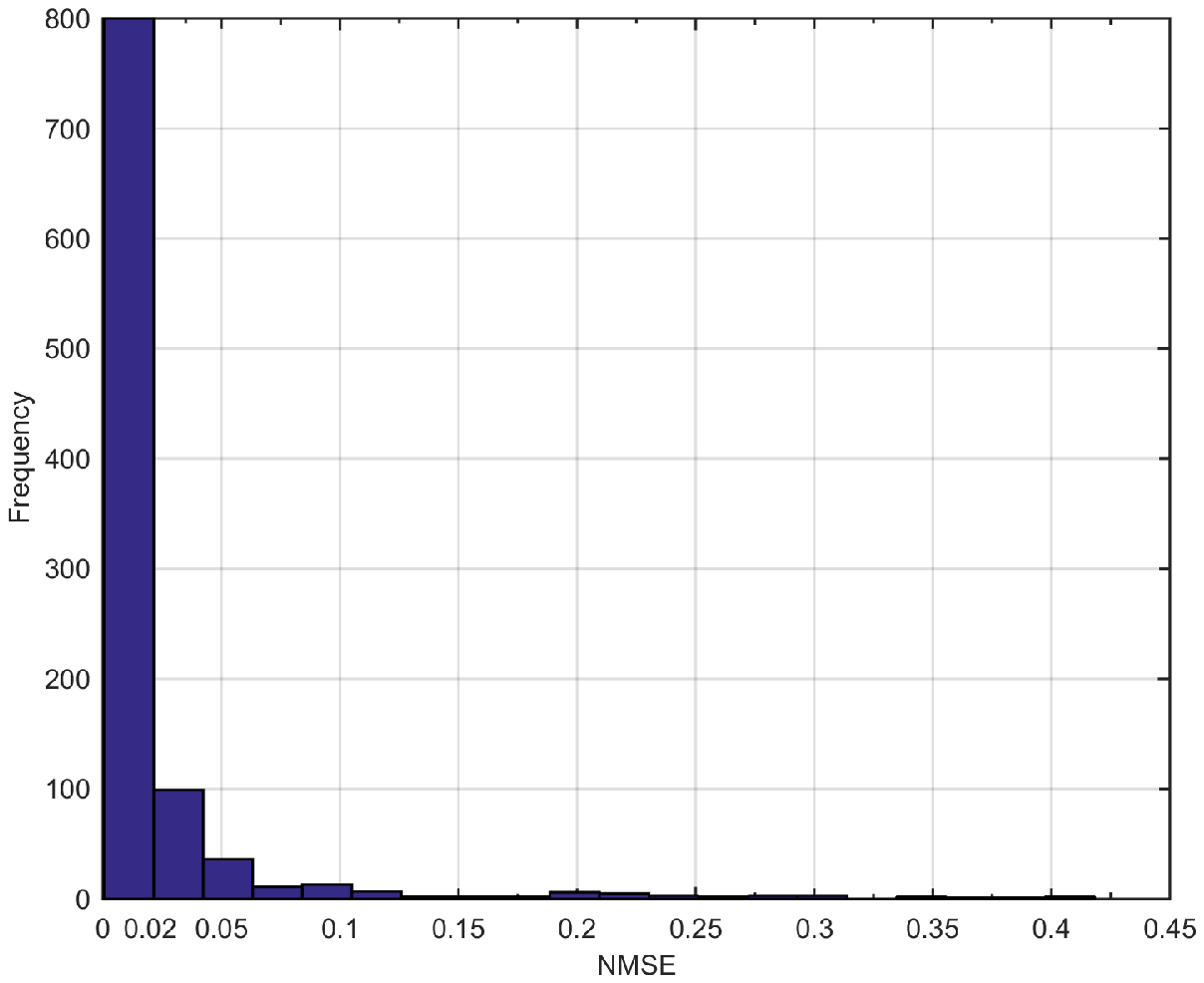}} \hfil
\subfigure[CS:FISTA]{\includegraphics[width=3.5in]{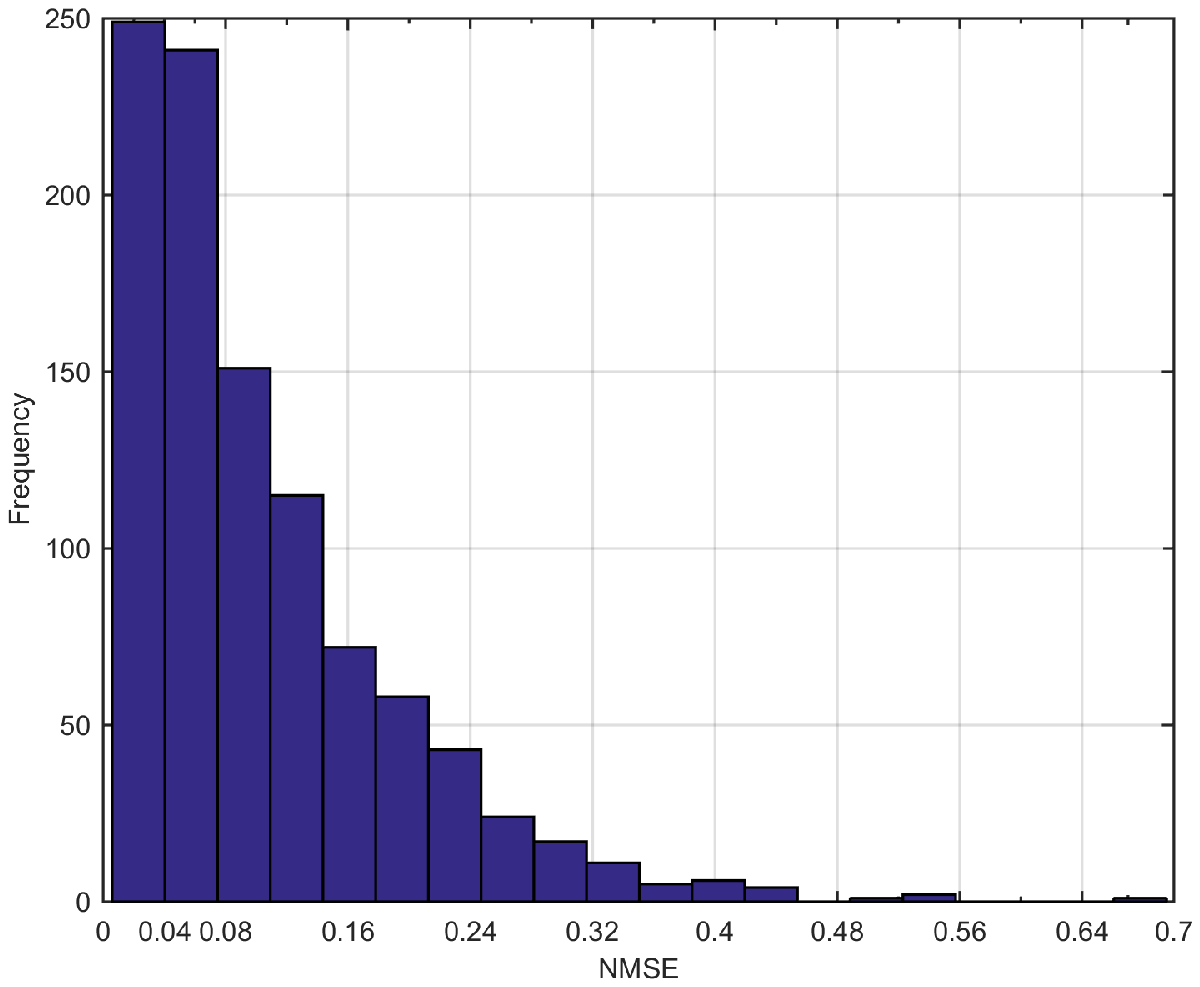}}
\caption{Histogram of the NMSE associated with the EC-MMSE
estimator and the compressed sensing method} \label{fig3}
\end{figure*}





Next, we evaluate the performance of the EC-MMSE estimator
proposed in Section \ref{sec:estimated-covariance-MMSE}. In our
simulations, channels are randomly generated according to the
one-ring model described above. Fig. \ref{fig2}(a) depicts the
NMSEs of respective methods as a function of the signal-to-noise
ratio (SNR), where we set $T=20$ and the SNR is defined as
$10\log(\|\boldsymbol{Xh}\|_2^2/T\sigma^2)$. Results are averaged
over 1000 independent runs, with the pilot sequence
$\boldsymbol{X}$ and the channel $\boldsymbol{h}$ randomly
generated for each run. In each run, the noise variance $\sigma^2$
is adjusted to meet a pre-specified SNR. A compressed sensing
method and a MMSE estimator which has access to the true
covariance matrix\footnote{The true covariance matrix is
calculated according to (\ref{eqn5}) via numerical average.} are
also included for comparison. For the compressed sensing method, a
fast iterative shrinkage-thresholding algorithm (FISTA) is
employed to estimate the channel based on (\ref{eqn4}). The
EC-MMSE estimator is built on the compressed sensing method: after
the virtual channel $\boldsymbol{\tilde{h}}$ is estimated via the
FISTA, we estimate the mean angle and the angular spread, then
obtain an estimate of the channel covariance matrix, and finally
construct the MMSE estimator. In our simulations, the angular
spread is estimated as as a symmetric interval around the
estimated mean angle, with $95\%$ of the channel energy
concentrated on the interval. From Fig. \ref{fig2}(a), we see that
our proposed scheme achieves a notably higher accuracy compared to
the compressed sensing method. This result shows that the
estimated covariance matrix, although imperfect, can still provide
a substantial performance improvement. Fig. \ref{fig2}(b) plots
the NMSEs of respective schemes vs. the number of symbols $T$,
where we set $\text{SNR}=20\text{dB}$. This result again
demonstrates the advantage of the proposed EC-MMSE estimator over
the compressed sensing method. To better illustrate the
performance, we plot the histogram in Fig. \ref{fig3} to show the
distribution of the NMSE for the EC-MMSE and the compressed
sensing method, respectively. From Fig. \ref{fig3}, we see that
the proposed EC-MMSE yields an accurate channel estimate (with an
NMSE within the range $[0,0.02]$) with a high probability, whereas
the NMSEs associated with the compressed sensing method spread
across the range $[0.04,0.4]$ with a high probability.


\begin{figure*}[!t]
 \subfigure[NMSE vs. SNR]{\includegraphics[width=3.5in]{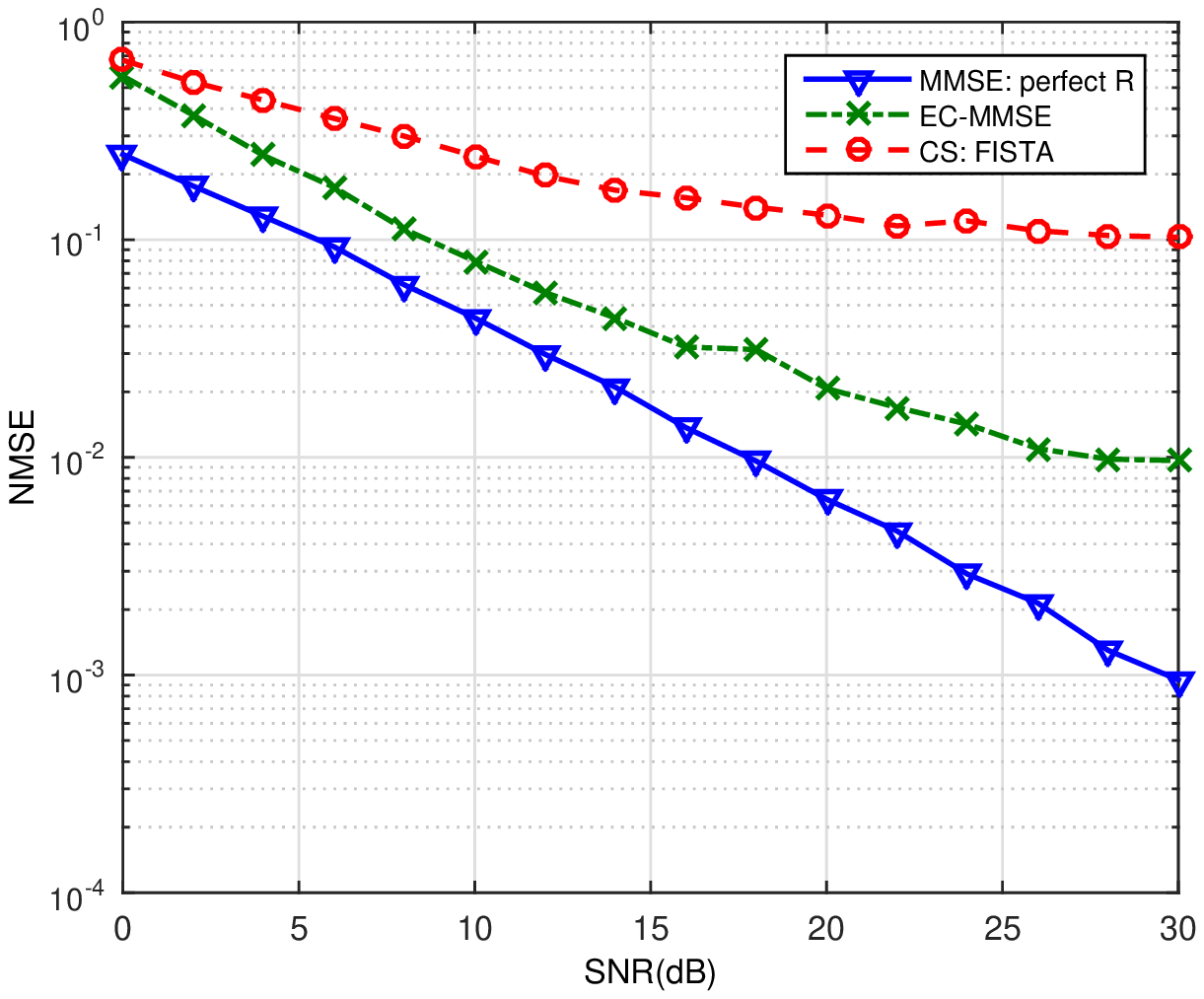}} \hfil
\subfigure[NMSE vs. number of symbols
$T$]{\includegraphics[width=3.5in]{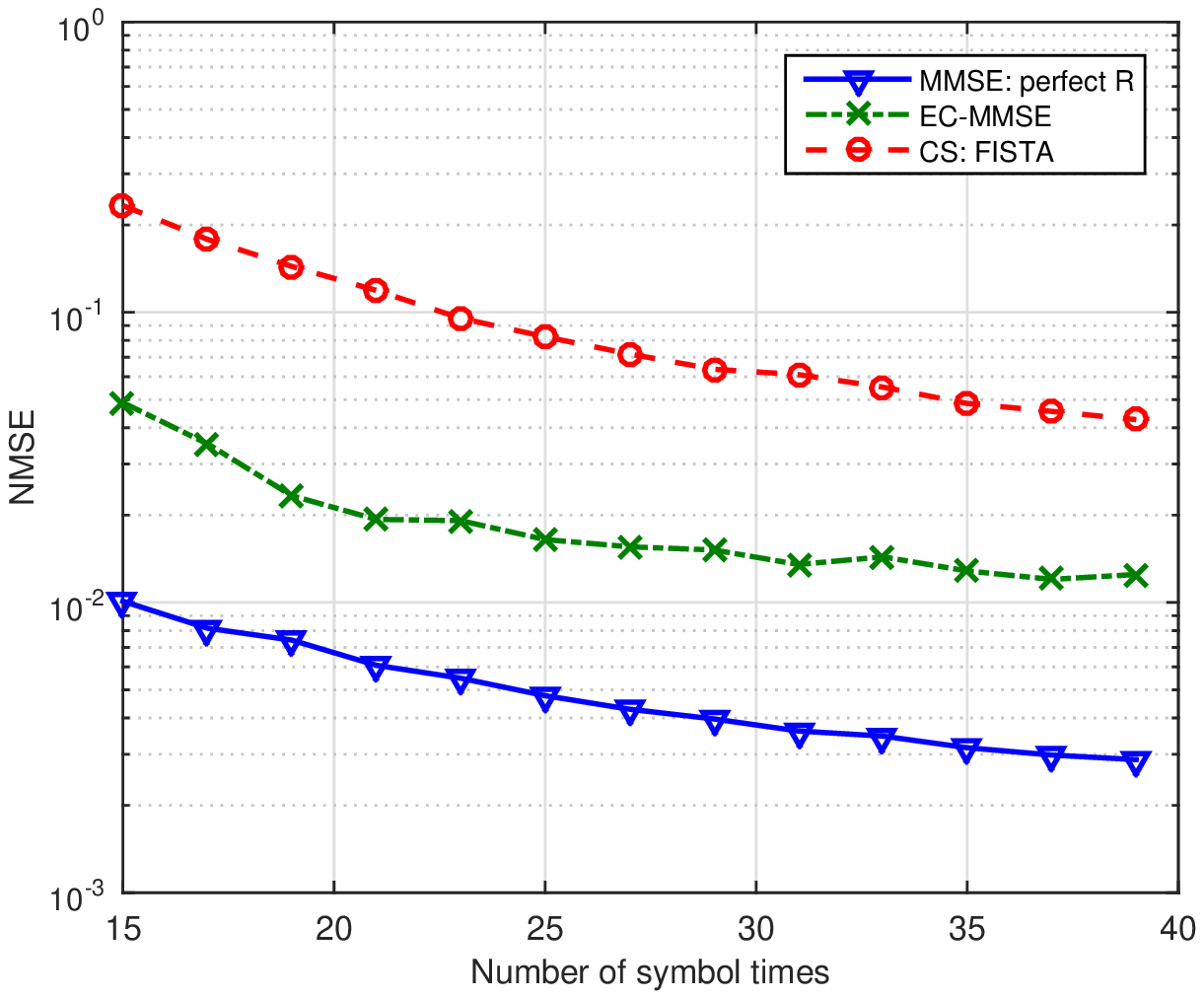}}
\caption{Gaussian AoA: NMSEs of respective schemes vs. SNR and
number of symbols $T$} \label{fig5}
\end{figure*}

Also, to examine the robustness of the proposed EC-MMSE estimator
against the model mismatch, in our simulations, we assume that the
angle of arrival follows a Gaussian distribution, with the mean
and the standard deviation set to be $\bar{\theta}=\pi/6$ and
$\sigma_{\theta}=\pi/30$, respectively. Note that in the proposed
EC-MMSE estimator, a uniform AoA distribution is assumed to
estimate the channel covariance matrix. In Fig. \ref{fig5}(a) and
Fig. \ref{fig5}(b), we plot the NMSEs of respective schemes as a
function of the SNR and the number of symbols, respectively, where
we set $T=20$ for Fig. \ref{fig5}(a) and $\text{SNR}=20\text{dB}$
for Fig. \ref{fig5}(b). Results are averaged over 1000 independent
runs, with the pilot sequence and the channel randomly generated
for each run. In each run, the noise variance is adjusted to meet
a pre-defined SNR. From Fig. \ref{fig5}, we see that the proposed
EC-MMSE estimator achieves superior performance even the presumed
AoA distribution is different from the true one. The reason, as
already explained in the previous section, is that the
eigenvectors of the channel covariance matrix are less dependent
on the AoA distribution. Therefore our scheme which assumes a
uniform AoA distribution can still reliably estimate the signal
subspace spanned by dominant eigenvectors, and as a result, the
EC-MMSE estimator still outperforms the compressed sensing method
by a big margin.

\begin{figure*}[!t]
 \subfigure[NMSE vs. the deviation (in angular degree) of the estimated mean angle from the true one]
 {\includegraphics[width=3.5in]{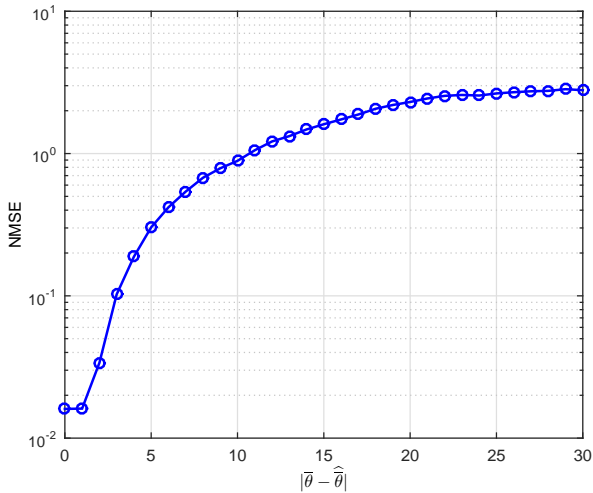}} \hfil
\subfigure[NMSE vs. the deviation (in angular degree) of the
estimated angle spread from the true
one]{\includegraphics[width=3.5in]{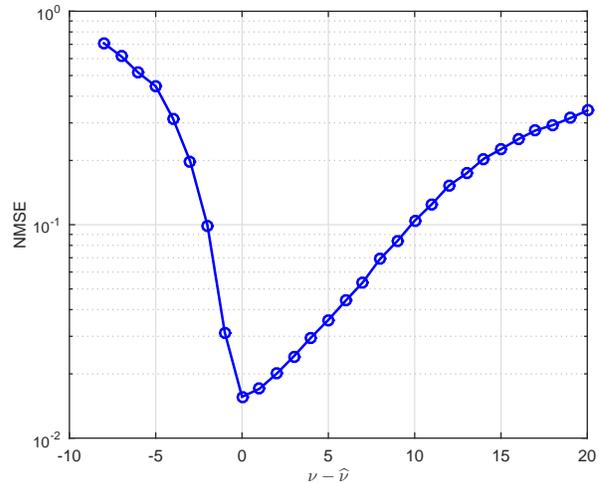}} \caption{NMSE of
EC-MMSE estimator vs. estimation errors of the mean angle and the
angular spread.} \label{fig4}
\end{figure*}

Lastly, to more thoroughly evaluate the performance of the
proposed EC-MMSE estimator, we examine its robustness against
estimation errors of the mean angle and the angular spread. Since
in the EC-MMSE scheme, the channel covariance matrix is obtained
based on the estimated mean angle and angular spread, estimation
errors of the mean angle and the angular spread will impair the
estimation quality of the covariance matrix, which, in turn,
affects the estimation accuracy of the EC-MMSE estimator. In Fig.
\ref{fig4}(a), we plot the NMSE of the EC-MMSE estimator as the
estimated mean angle deviates from the true one, where we set
$T=15$, $\sigma^2=0.1$, and the angular spread is assumed
perfectly known. Results are averaged over $10^3$ independent
runs, and for each run, the pilot sequence is randomly generated
to meet a pre-specified power constraint, and the channel is
randomly generated according to the one-ring model described in
the second paragraph of this section. We see that the EC-MMSE
estimator exhibits some robustness against the mean angle
mismatch: the EC-MMSE estimator incurs mild performance
degradation if the deviation of the estimated mean angle from the
true one is small, say,
$|\bar{\theta}-\hat{\bar{\theta}}|<3^{\circ}$. Nevertheless, a
large deviation would result in a significant performance
degradation. Fig. \ref{fig4}(b) depicts the behavior of the
proposed EC-MMSE estimator when the estimated angular spread
deviates from the true angular spread, where the mean angle is
assumed perfectly estimated. From Fig. \ref{fig4}(b), it can be
observed that the EC-MMSE estimator is robust to an overestimation
of the angular spread, but is sensitive to the underestimation
errors: it suffers from a substantial performance loss when the
estimated angular spread is smaller than the true one. Hence it is
safer to overestimate than to underestimate the angular spread.


\section{Conclusions} \label{sec:conclusion}
We considered the problem of downlink training and channel
estimation for FDD massive MIMO systems. Since the required amount
of overhead for downlink training grows linearly with the number
of transmit antennas at the BS, reducing the overhead for downlink
training and uplink feedback has been a central issue in FDD
massive MIMO systems. In this paper, we exploited the low-rank
structure of the channel covariance matrix to reduce the overhead
for downlink training. We studied the asymptotic behavior of the
MMSE estimator when the channel covariance matrix has a low-rank
structure. Our analysis shows that the MMSE estimator can achieve
an exact channel recovery in the asymptotic low-noise regime,
provided that the number of pilot symbols in time is no smaller
than the rank of the channel covariance matrix. We also examined
the optimal pilot sequence design for the single-user case, and an
asymptotic optimal pilot sequence design for the multi-user
scenario. We also develop a training-free scheme to estimate the
channel covariance matrix. Simulation results show that a MMSE
estimator based on the estimated covariance matrix achieves a
substantial performance improvement as compared with the
compressed sensing method, and is robust against the AoA
distribution mismatch and the angular spread estimation error.

\bibliography{newbib}
\bibliographystyle{IEEEtran}

\end{document}